\documentclass[envcountsame]{llncs}

\usepackage{amsmath}
\usepackage{amssymb}
\usepackage{mathtools}
\usepackage{hyperref}
\usepackage[inline]{enumitem}
\usepackage{array}
\usepackage{balance}
\usepackage{color}
\usepackage{algorithm}
\usepackage{xfrac}
\usepackage[noend]{algpseudocode}
\usepackage{xspace}
\usepackage{subcaption}

\let\doendproof\endproof
\renewcommand\endproof{\qed\doendproof}

\let\set\mathbb

\newcommand{\K}{\set{K}}
\newcommand{\N}{\set{N}}

\newcommand{\Z}{\set{Z}}
\newcommand{\KX}{\set{K}[\fls{x}]}

\newcommand{\ass}{\mathrel{:=}}     
\newcommand{\reserved}[1]{\textbf{#1}} 
\newcommand{\DO}{\reserved{do}}
\newcommand{\OD}{\reserved{end~while}}
\newcommand{\WHILE}{\reserved{while}}
\newcommand{\IF}{\reserved{if}}
\newcommand{\FI}{\reserved{end~if}}
\newcommand{\THEN}{\reserved{then}}

\newcommand{\ELSE}{\reserved{else}}
\newcommand{\I}{\mathcal{I}}

\newcommand{\hg}{\operatorname{hg}}

\newcommand{\fn}[1]{(#1)^{\underline{n}}}

\newcommand{\ideal}{\vartriangleleft}

\newcommand{\fls}[1]{\vec{#1}} 

\newcommand{\Aligator}{\textsc{Aligator}\xspace}

\algrenewcommand\algorithmicrequire{\textbf{Input:}}
\algrenewcommand\algorithmicensure{\textbf{Output:}}

\algnewcommand\algorithmicand{\textbf{\textsf{AND}}}
\algrenewcommand\algorithmicwhile{\textbf{\textsf{WHILE}}}
\algrenewcommand\algorithmicfor{\textbf{\textsf{FOR}}}
\algrenewcommand\algorithmicdo{\textbf{\textsf{DO}}}
\algrenewcommand\algorithmicreturn{\textbf{\textsf{RETURN}}}

\newcommand{\algrule}[1][.2pt]{\par\vskip.5\baselineskip\hrule height #1\par\vskip.5\baselineskip}

\expandafter\def\expandafter\normalsize\expandafter{%
\normalsize
\setlength\abovedisplayskip{5pt}
\setlength\belowdisplayskip{6pt}
\setlength\abovedisplayshortskip{6pt}
\setlength\belowdisplayshortskip{6pt}
}

\title{
Invariant Generation for Multi-Path Loops with Polynomial Assignments}

\author{Andreas Humenberger \and Maximilian Jaroschek \and Laura Kov\'acs%
\thanks{All authors are supported by the ERC Starting Grant 2014 SYMCAR
639270. Furthermore, we acknowledge funding from the Wallenberg Academy
Fellowship 2014 TheProSE, the Swedish VR grant GenPro D0497701, and the Austrian
FWF research project RiSE S11409-N23. We also acknowledge support from the FWF
project W1255-N23.}}

\institute{Technische Universit\"at Wien\\
Institut f\"ur Informationssysteme 184\\
Favoritenstra{\ss}e 9-11\\
Vienna A-1040, Austria\\
\email{ahumenbe@forsyte.at}\\
\email{maximilian@mjaroschek.com}\\
\email{lkovacs@forsyte.at}}

\begin{document}

\urlstyle{tt}

\maketitle

\begin{abstract}
Program analysis requires the generation of program properties expressing
conditions to hold at intermediate program locations. When it comes to programs
with loops, these properties are typically expressed as loop invariants. In this
paper we study a class of multi-path program loops with numeric variables, in
particular nested loops with conditionals, where assignments to program
variables are polynomial expressions over program variables. We call this class
of loops {\it extended P-solvable} and introduce an algorithm for generating all
polynomial invariants of such loops. By an iterative procedure employing
Gr\"obner basis computation, our approach computes the polynomial ideal of the
polynomial invariants of each program path and combines these ideals
sequentially until a fixed point is reached. This fixed point represents the
polynomial ideal of all polynomial invariants of the given extended P-solvable
loop. We prove termination of our method and show that the maximal number of
iterations for reaching the fixed point depends linearly on the number of
program variables and the number of inner loops. In particular, for a loop with
$m$ program variables and $r$ conditional branches we prove an upper bound of
$m\cdot r$ iterations. We implemented our approach in the \Aligator{} software
package. Furthermore, we evaluated it on 18 programs with polynomial arithmetic
and compared it to existing methods in invariant generation. The results show
the efficiency of our approach.
\end{abstract}


\section{Introduction}


Reasoning about programs with loops requires loop invariants expressing
properties that hold before and after every loop iteration. The difficulty of
generating such properties automatically comes from the use of non-linear
arithmetic, unbounded data structures, complex control flow, just to name few of
the reasons. In this paper we focus on multi-path loops with numeric variables
and polynomial arithmetic and introduce an automated approach inferring {\it
all} loop invariants as polynomial equalities among program variables. For doing
so,   we identify a class of multi-path loops with nested conditionals, where
assignments to program variables are polynomial expressions over program
variables. Based on our previous work~\cite{issac2017}, we call this class of
loops \emph{extended P-solvable}. Compared to~\cite{issac2017} where only
single-path programs with polynomial arithmetic were treated, in this paper we
generalize the notion of  {extended  P-solvable} loops to multi-path loops;
single-path loops being thus a special case of our method. 

For the class of extended P-solvable loops, we introduce an automated approach
computing all polynomial invariants. Our work exploits the results
of~\cite{kapur,reasoningalgebraically} showing that the set of polynomial
invariants forms a polynomial ideal, called the polynomial invariant ideal.
Hence, the task of generating all polynomial invariants reduces to the problem
of generating a basis of the polynomial invariant ideal. Following this
observation, given an extended P-solvable loop with nested conditionals, we
proceed as follows: we (i) turn the  multi-path loop into a sequence of
single-path loops, (ii) generate the polynomial invariant ideal of each
single-path loop and (iii) combine these ideals iteratively until the polynomial
invariant ideal of the multi-path loop is derived. 

A crucial property of extended P-solvable loops is that the single-path loops
corresponding to one path of the multi-path loop are also extended P-solvable.
For generating the polynomial invariant ideal of extended P-solvable single-path
loops, we model loops by a system of algebraic recurrences, compute the closed
forms of these recurrences by symbolic computation as described
in~\cite{issac2017} and compute the Gr\"obner basis of the polynomial invariant
ideal from the system of closed forms. When combining the polynomial invariant
ideals of each extended P-solvable single-path loop, we prove that the
``composition'' maintains the properties of extended P-solvable loops. Further,
by exploiting the algebraic structures of the polynomial invariant ideals of
extended P-solvable loops, we prove that the process of iteratively combining
the polynomial invariant ideals of each extended P-solvable single-path loop is
finite. That is, a fixed point is reached in a finite number of steps. We prove
that this fixed point is the polynomial invariant ideal of the extended
P-solvable loop with nested conditionals. We also show that reaching the fixed
point depends linearly on the number of program variables and the number of
inner loops. In particular, for a loop with $m$ program variables and $r$ inner
loops (paths) we prove an upper bound of $m\cdot r$ iterations. The termination
proof of our method implies the completeness of our approach: for an extended
P-solvable loop with nested conditionals, our method computes all its polynomial
invariants. This result generalizes and corrects the result
of~\cite{completeinvariant} on programs for more restricted arithmetic than
extended P-solvable loops. Our class of programs extends the programming model
of~\cite{completeinvariant} with richer arithmetic and our invariant generation
procedure also applies to~\cite{completeinvariant}. As such, our proof of
termination also yields a termination proof for~\cite{completeinvariant}. 

We implemented our approach in the open source Mathematica package \Aligator{}
and evaluated our method on 18 challenging examples. When compared to
state-of-the-art tools in invariant generation, 
\Aligator{} performed much better in 14 examples out of 18. 

The paper is organized as follows: We start by giving the necessary details
about our programming model in Section~\ref{sec:programming-model} and provide
background about polynomial rings and ideals in Section~\ref{sec:ideals}. In
Section~\ref{sec:loopass} we recall the notion of extended P-solvable loops
from~\cite{issac2017}. The lemmas and propositions of
Section~\ref{sec:dependencies} will then help us to prove termination of our
invariant generation procedure in Section~\ref{sec:loopcond}.
Finally, Section~\ref{sec:implementation} describes our implementation in \Aligator, together with an experimental evaluation of our
approach. \\

\noindent{\bf Related Work.} 
Generation of non-linear loop invariants 
has been addressed in previous research. We discuss here some of the
most related works that we are aware of. 

The methods of \cite{olm,san} compute polynomial equality invariants
by fixing an a priori 
bound on the degree of the polynomials. Using this bound, a
template invariant of fixed degree is constructed. Properties of
polynomial invariants, e.g. inductiveness, are used  to generate
constraints over the unknown coefficients of the template coefficients
and these constraints are then solved in linear or polynomial
algebra. An a priori fixed polynomial degree is also used
in~\cite{Carbonell07a,farewellgroebner}. Unlike these approaches, in our work we do not
fix the degree of polynomial invariants but generate all polynomial
invariants (and not just invariants up to a fixed degree). Our
restrictions come in the programming model, namely treating only loops
with nested conditionals and polynomial arithmetic. For such programs,
our approach is complete. 

Another line of research uses abstract interpretation in conjunction
with recurrence solving and/or polynomial algebra. The work
of~\cite{kapur} generates all polynomial invariants of so-called {\it
  simple loops} with
nested conditionals. The approach combines abstract interpretation
with polynomial ideal theory. Our model of extended P-solvable loops
is much more general than simple loops, for example we allow
multiplication with the loop counter and treat 
algebraic, and not only rational, numbers in closed form solutions. 
Abstract interpretation is also used in~\cite{farzan,oliveira,kincaidPOPL18} to
infer non-linear invariants. The programming model of these works
handle loops whose assignments induce
linear recurrences with constant coefficients. Extended P-solvable
loops can however yield more complex recurrence equations.  In
particular, when comparing our work to~\cite{kincaidPOPL18}, we note
that the recurrence equations of program variables
in~\cite{kincaidPOPL18} correspond to a subclass of linear recurrences
with constant coefficients: namely, recurrences whose closed form
representations do not include non-rational algebraic numbers. Our work
treats the entire class of linear recurrences with constant
coefficients and even handles programs whose arithmetic operations
induce a class of linear recurrences with polynomial coefficients in
the loop counter. While the non-linear arithmetic of our work is more
general than the one in~\cite{kincaidPOPL18}, we note that the
programming model of~\cite{kincaidPOPL18} can handle programs that are
more complex than the ones treated in our work, in particular due to
the presence of nested loops and function/procedure calls. Further,
the invariant generation approach of~~\cite{kincaidPOPL18} is
property-guided: invariants are generated in order to prove the safety
assertion of the program. Contrarily to this, we generate all
invariants of the program and not only the ones implying the safety
assertion.

Solving recurrences and computing polynomial invariant ideals from a
system of closed form solution is also described
in~\cite{reasoningalgebraically}. Our work builds upon the results
of~\cite{reasoningalgebraically} but
generalizes~\cite{reasoningalgebraically} to extended P-solvable
loops. Moreover, we also prove that our invariant generation
procedure terminates. Our termination result generalizes~\cite{completeinvariant} by
handling programs with more complex polynomial
arithmetic. Furthermore, 
instead of computing the
invariant ideals of all permutations of a given set of inner loops and extending
this set until a polynomial ideal as a fixed point is reached, we
generate the polynomial invariant ideal of just one
permutation iteratively until we reach the fixed point. As a result we have to
perform less Gr\"obner basis computations in the process of invariant generation.

A data-driven approach to invariant generation is given in~\cite{datadriveninvariants}, where concrete
program executions are used to generate invariant candidates. Machine learning
is then used to infer polynomial invariants from the candidate
ones. In our work we do not use invariant candidates. While the program
flow in our 
programming model is more restricted then~\cite{datadriveninvariants}, to the best of our knowledge, none of the
above cited methods can fully handle the polynomial arithmetic of extended
P-solvable loops.



\section{Preliminaries}
\label{sec:prelims}

\subsection{Programming Model and Invariants}\label{sec:programming-model}
Let $\set K$ be a computable field of characteristic zero. This means that
addition and multiplication can be carried out algorithmically, that there
exists an algorithm to test if an element in $\set K$ is zero, and that the
field of rational numbers~$\set Q$ is a subfield of $\set K$. For variables
$x_1,\dots,x_n$, the ring of multivariate polynomials over $\set K$ is denoted
by $\set K[x_1,\dots,x_n]$, or, if the number of variables is clear from (or
irrelevant in) the context, by $\KX$. Correspondingly, $\set K(x_1,\dots,x_m)$
or $\set K(\fls{x})$ denotes the field of rational functions over $\set K$ in
$x_1,\dots,x_m$. If every polynomial in $\set K[x]$ with
a degree $\geq 1$ has at least one root in $\set K$, then $\set K$ is called
algebraically closed. An example for such a field is $\overline{\set Q}$, the
field of algebraic numbers. In contrast, the field of complex numbers $\set C$
is algebraically closed, but not computable, and $\set Q$ is computable, but not
algebraically closed. We suppose that $\set K$ is always algebraically
closed. This is not necessary for our theory, as we only need the existence of
roots for certain polynomials, which is achieved by choosing~$\set K$ to be an
appropriate algebraic extension field of $\set Q$. It does, however, greatly
simplify the statement of our results.

In our framework, we consider a program $B$ to be a loop of the form 
\begin{equation}
\label{eq:program}
    \begin{tabular}{l}
     \WHILE\ \dots\ \DO\\
      \quad $B'$\\
      \OD
    \end{tabular}
\end{equation}
    where $B'$ is a program block that is either the empty block~$\epsilon$, an
    assignment $v_i=f(v_1,\dots,v_m)$ for a rational function
    $f\in\set K(x_1,\dots,x_m)$ and program variables $v_1,\dots,v_m$, or has
    one of the composite forms 
\begin{equation*}
    \setlength{\tabcolsep}{12pt}
    \begin{tabular}{l|l|l}
      \text{sequential} & \hspace{4px}\text{inner loop} & \hspace{3px}\text{conditional}\\
      & & \\
      \ & \WHILE\ \dots\ \DO & \IF\ \dots\ \THEN\\
      \hspace{5px}$B_1; B_2$ & \quad $B_1$ & \quad $B_1$\ \ELSE\ $B_2$\\
      \ & \OD & \FI
    \end{tabular}
\end{equation*}
for some program blocks $B_1$ and $B_2$ and the usual semantics. We omit
conditions for the loop and if statements, as the problem of computing all
polynomial invariants is undecidable when taking affine equality tests into
account~\cite{olm}. Consequently, we regard loops as non-deterministic programs
in which each block of consecutive assignments can be executed
arbitrarily often. More precisely, grouping consecutive assignments into blocks
$B_1,\dots,B_r$, any execution path of $B$ can be written in the form
\[B_1^{n_1};B_2^{n_2};\dots;B_r^{n_r};B_1^{n_{r+1}};B_2^{n_{r+2}};\dots\]
for a sequence $(n_i)_{i\in\set N}$ of non-negative integers with finitely many
non-zero elements. To that effect, we interpret any given
program~\eqref{eq:program} as the set of its execution paths, written as
{ \setlength\abovedisplayskip{2pt}%
\setlength\belowdisplayskip{2pt}%
\[B=(B_1^*;B_2^*;\dots;B_r^*)^*.\]
}

We adapt the well-established Hoare triple notation
\begin{equation}
\label{eq:hoare}
\{P\}B\{Q\},
\end{equation}
for program specifications, where $P$ and $Q$ are logical formulas, called the
pre- and postcondition respectively, and $B$ is a program. In this paper we
focus on partial correctness of programs, that is a Hoare triple~\eqref{eq:hoare}
is correct if every terminating computation of $B$ which starts in a state
satisfying $P$ terminates in a state that satisfies $Q$.

In this paper we are concerned with computing polynomial invariants for a
considerable subset of loops of the form~\eqref{eq:program}. These invariants
are algebraic dependencies among the loop variables that hold after any number
of loop iterations.

\begin{definition}
  A polynomial $p\in\set K[x_1,\dots,x_m]$ is a polynomial loop invariant for a
  loop $B=B_1^*;\dots;B_r^*$ in the program variables $v_1,\dots,v_m$ with initial values
  $v_1(0),\dots,v_m(0)$, if for every sequence $(n_i)_{i\in\set N}$ of
  non-negative integers with finitely many non-zero elements, the Hoare
  triple
{ \setlength\abovedisplayskip{1pt}%
  \setlength\belowdisplayskip{1pt}%
  \begin{align*}
  & \{p(v_1,\dots,v_m)=0\wedge\bigwedge_{i=0}^m v_i=v_i(0)\}\\
  &\quad B_1^{n_1};B_2^{n_2}\dots,B_r^{n_r};B_1^{n_{r+1}};\dots\\[5px]
  & \{p(v_1,\dots,v_m)=0\}
  \end{align*}
}%
is correct.
\end{definition}


\subsection{Polynomial Rings and Ideals}
\label{sec:ideals}
 
Polynomial invariants are algebraic dependencies among the values of the
variables at each loop iteration. Obviously, non-trivial dependencies do not
always exist.

\begin{definition}
  Let $\set L\mathbin{/}\K$ be a field extension. Then ${a_1,\dots,a_n \in\set L}$
  are \emph{algebraically dependent} over $\K$ if there exists a $p \in
  \K[x_1,\dots,x_n]\setminus\{0\}$ such that $p(a_1,\dots,a_n) = 0$. Otherwise
  they are called \emph{algebraically independent}.
\end{definition}

In~\cite{laura,kapur}, it is observed that the set of all polynomial loop
invariants for a given loop forms an ideal.  It is this fact that facilitates
all of our subsequent reasoning.
\begin{definition}
  A subset $\I$ of a commutative ring $R$ is called an \emph{ideal}, written
  $\I\vartriangleleft R$, if it satisfies the following
  three properties:
  \begin{enumerate}
  \item $0\in\I$.
  \item For all $a,b\in\I$: $a+b\in\I$.
  \item For all $a\in\I$ and $b\in R$: $a\cdot b\in\I$.
  \end{enumerate}
\end{definition}

\goodbreak

\begin{definition}
  Let $\I\ideal R$. Then $\I$ is called
  \begin{itemize}
    \item \emph{proper} if it is not equal to $R$,
    \item \emph{prime} if $a \cdot b \in\I$ implies $a\in\I$ or $b\in\I$, and
    \item \emph{radical} if $a^n\in\I$ implies $a\in\I$.
  \end{itemize}
  The \emph{height} $\hg(\I)\in\set N$ of a prime ideal $\I$ is equal to $n$ if
  $n$ is the maximal length of all possible chains of prime ideals $\I_0\subset
  \I_2\subset\dots\subset \I_n=\I$.
\end{definition}

\begin{example}
  The set of even integers $2\Z$ is an ideal of $\Z$. In general $n\Z$ for a
  fixed integer $n$ is an ideal of $\Z$. It is prime if and only if $n$ is a
  prime number.
\end{example}

Polynomial ideals can informally be interpreted as the set of all consequences
when it is known that certain polynomial equations hold. In fact, if we have
given a set $P$ of polynomials of which we know that they serve as algebraic
dependencies among the variables of a given loop, the ideal generated by~$P$
then contains all the polynomials that consequently have to be polynomial
invariants as well.

\begin{definition}
A subset $B\subseteq\I$ of an ideal $\I\vartriangleleft R$ is called a \emph{basis} for
$\I$ if 
\begin{equation*}
\I=\langle B\rangle :=\{a_0b_0+\dots+a_mb_m\mid m\in\set N,
  a_0,\dots,a_m\in R,b_0\dots,b_m\in B\}.
\end{equation*}
We say that $B$ \emph{generates} $\I$.
\end{definition}

A basis for a given ideal in a ring does not necessarily have to be finite.
However, a key result in commutative algebra makes sure that in our setting we
only have to consider finitely generated ideals.

\begin{theorem}[Hilbert's Basis Theorem -- Special case]
  Every ideal in $\KX$ has a finite basis. 
\end{theorem}

Subsequently, whenever we say we are given an ideal $\I$, we mean that we have
given a finite basis of $\I$.  

There is usually more than one basis for a given ideal and some are more useful
for certain purposes than others. In his seminal PhD thesis~\cite{Buchberger06},
Buchberger introduced the notion of Gr\"obner bases for polynomial ideals and an
algorithm to compute them. While, for reasons of brevity, we will not formally
define these bases, it is important to note that with their help, central
questions concerning polynomial ideals can be answered algorithmically.

\begin{theorem}
\label{thm:gb}
  Let $p\in\set K[x_1,\dots,x_n]$ and $\I,\mathcal{J}\vartriangleleft\set
  K[x_1,\dots,x_n]$. There exist algorithms to decide the following problems.
  \begin{enumerate}
  \item Decide if $p$ is an element of $\I$.
  \item Compute a basis of $\I+ \mathcal{J}$.
  \item Compute a basis of $\I\cap \mathcal{J}$.
  \item For $\{\tilde{x}_1,\dots\tilde{x}_m\}\subseteq\{x_1,\dots,x_n\}$,
  compute a basis of $\I\cap \set K[\tilde{x}_1,\dots,\tilde{x}_m]$.
\item Let $q\in\KX$. Compute a basis for 
{ \setlength\abovedisplayskip{5pt plus 2pt minus 2pt}%
  \setlength\belowdisplayskip{5pt plus 2pt minus 2pt}%
\[\I\mathbin{:}\langle q\rangle^\infty:=\{q\in\KX\mid \exists n\in\set N:
  q^np\in\I\}.\] 
}
The ideal $\I\mathbin{:}\langle q\rangle^\infty\kern-1pt$ is called the
  \emph{saturation} of $\I$ with respect to $q$.
  \end{enumerate}
\end{theorem}

We will use Gr\"obner bases to compute the ideal of all algebraic relations
among given rational functions. For this, we use the polynomials $q_iy_i-p_i$ to model
the equations $y_i=q_i/p_i$ by multiplying the equation with the denominator. In
order to model the fact that the denominator is not identically zero, and
therefore allowing us to divide by it again, we use the saturation with respect
to the least common multiple of all denominators. To see why this is necessary,
consider $y_1=y_2=\frac{x_1}{x_2}$. An algebraic relation among $y_1$ and $y_2$
is $y_1-y_2$, but with the polynomials $x_2y_1-x_1$ and $x_2y_2-x_1$, we
only can derive $x_2(y_1-y_2)$. We have to divide by $x_2$.

\begin{theorem}
  Let $r_1,\dots,r_m\in\set K(\fls{x})$ and let the numerator of $r_i\kern-2pt$ be given
  by $p_i\in\KX$ and the denominator by $q_i\in\KX$. The ideal of all
  polynomials $p$ in $\set K[\fls{y}]$ with $p(r_1,\dots,r_m)=0$ is given by
\[\biggl(\smash{\sum_{i=1}^m}\langle
  q_iy_i-p_i\rangle\biggr)\mathbin{:}\langle\operatorname{lcm}(q_1,\dots,q_m)\rangle^\infty\cap
  \set K[\fls{y}],\]
where $\operatorname{lcm}(\dots)$ denotes the least common multiple.
\end{theorem}

\begin{proof}
  Write $d:=\operatorname{lcm}(q_1,\dots,q_m)$. The theorem can be easily verified from the fact that, for any given $p$ with
  $p(r_1,\dots,r_m)=0$, there exists a $k\in\set N$ such that $d^kp(r_1,\dots,r_m)=0$ is an
  algebraic relation for $p_1,\dots,p_m$ (by clearing denominators in the
  equation $p(r_1,\dots,r_m)=0).$ 
\end{proof}

A polynomial ideal $\I\vartriangleleft \KX$ gives rise to a set of points
in~$\set K^n$ for which all polynomials in $\I$ vanish simultaneously. This set
is called a \emph{variety}.

\begin{definition}
  Let $\I \ideal \K[x_1,\dots,x_n]$ be an ideal. The set
{ \setlength\abovedisplayskip{5pt plus 2pt minus 2pt}%
  \setlength\belowdisplayskip{5pt plus 2pt minus 2pt}%
  \[ V(\I) = \{ (a_1,\dots,a_n) \in \K^n \mid p(a_1,\dots,a_n) = 0~\text{for
      all}~p\in\I \},\]
}is the \emph{variety} defined by $\I$.
\end{definition}

Varieties are one of the central objects of study in algebraic geometry. Certain
geometric shapes like points, lines, circles or balls can be described by prime
ideals and come with an intuitive notion of a dimension, e.g.\ points have
dimension zero, lines and circles have dimension one and balls have dimension
two. The notion of the Krull dimension of a ring formalizes this intuition when
being applied to the quotient ring $\KX/\I$. In this paper, we will use the
Krull dimension to provide an upper bound for the number of necessary iterations
of our algorithm.

\begin{definition}
  The \emph{Krull dimension} of a commutative ring $R$ is the supremum of
  the lengths of all chains $\I_0\subset\I_1\subset\dots $ of prime ideals.
\end{definition}

\begin{theorem}\label{thm:finitedimension}
  The Krull dimension of $\set K[x_1,\dots,x_n]$ is equal to $n$.
\end{theorem}


\section{Extended P-Solvable Loops}

In~\cite{issac2017} the class of \emph{P-solvable}
loops~\cite{reasoningalgebraically} was extended to so-called \emph{extended
P-solvable} loops. So far, this class captures loops with assignments only,
i.e.~loops without any nesting of conditionals and loops. In
Section~\ref{sec:loopcond} we close this gap by introducing a new approach for
computing invariants of multi-path loops which generalizes the algorithm
proposed in~\cite{completeinvariant}. Before dealing with multi-path loops, we
recall the notion of extended P-solvable loops in Section~\ref{sec:loopass} and
showcase the invariant ideal computation.

\subsection{Loops with assignments only}\label{sec:loopass}

In this section, we restrain ourselves to loops whose bodies are comprised of
rational function assignments only. This means that we restrict the valid
composite forms in a program of the form~\eqref{eq:program} to sequential
compositions and, for the moment, exclude inner loops and conditional branches.
We therefore consider a loop $L=B_1^*$ where $B_1$ is a single block containing
only variable assignments.

Each variable $v_i$ in a given loop of the form~\eqref{eq:program} gives rise to
a sequence $(v_i(n))_{n\in\N}$, where $n$ is the number of loop iterations. The
class of eligible loops is then defined based on the form of these sequences.
Let $r\fn{x}$ denote the \emph{falling factorial} defined as
${\prod_{i=0}^{n-1}r(x-i)}$ for any $r\in\K(x)$ and $n\in\N$.

\begin{definition}\label{def:extendedpsolvable}
  A loop with assignments only is called \emph{extended P-solvable} if each of
  its recursively changed variables determines a sequence of the form
  %
  \begin{equation}
    \label{eq:psolvable}
    v_i(n)=\sum_{j\in\set
      Z^\ell}p_{i,j}(n,\theta_1^n,\dots,\theta_k^n)((n+\zeta_1)^{\underline n})^{j_1}\cdots
    ((n+\zeta_\ell)^{\underline n})^{j_\ell}
  \end{equation}
  %
  where $k,\ell\in\set N$, the $p_{i,j}$ are polynomials in
$\set K(x)[y_1,\dots,y_k]$, not identically zero for finitely many
$j\in\set Z^\ell$, the $\theta_i$ are elements of $\set K$ and the $\zeta_i$ are
elements of $\set K\setminus\set Z^-$ with $\theta_i\neq\theta_j$ and
$\zeta_i-\zeta_j\notin\set Z$ for $i\neq j$.
\end{definition}

Definition~\ref{def:extendedpsolvable} extends the class of P-solvable loops in
the sense that each sequence induced by an extended P-solvable loop is the sum
of a finitely many hypergeometric sequences. This comprises C-finite sequences
as well as hypergeometric sequences and sums and Hadamard products of C-finite
and hypergeometric sequences. In contrast, P-solvable loops induce C-finite
sequences only. For details on C-finite and hypergeometric sequences we refer
to~\cite{kauers}.

Every sequence of the form~\eqref{eq:psolvable} can be written as
{ \setlength\abovedisplayskip{5pt plus 2pt minus 2pt}%
  \setlength\belowdisplayskip{5pt plus 2pt minus 2pt}%
\[ \smash{v^{(1)}_j}\kern-2pt = \smash{r_j(\fls{v}^{(0)}}\kern-2pt,\fls{\theta},(n + \fls{\zeta})^{\underline n}, n) \]
}where $r_j=p_i/q_i$ is a rational function, and $v^{(0)}$ and $v^{(1)}$ denote
the values of $v$ before and after the execution of the loop. Let
$I(\fls{\theta},\fls{\zeta})\ideal \K[y_0,\dots,y_{k+\ell}]$ be the ideal of all
algebraic dependencies in the variables $y_0,\dots,y_{k+\ell}$ between the
sequence $(n)_{n\in\set N}$, the exponential sequences
$\theta_1^n,\dots,\theta_k^n$ and the sequences
$(n+\zeta_1)^{\underline{n}},\dots,(n+\zeta_\ell)^{\underline{n}}$. Note that it
was shown in~\cite{issac2017} that this ideal is the same as the extension of
the ideal $I(\fls{\theta})\vartriangleleft\set K[y_0,\dots,y_k]$ of all algebraic dependencies between the $\theta^n\kern-2pt$ in
$\set K[y_0,\dots,y_k]$ to $\set K[y_0,\dots,y_{k+\ell}]$, as the factorial
sequences $(n+\zeta_i)^{\underline{n}}$ are algebraically independent from the
exponential sequences $\theta_i^n$. Now the following proposition states how the
invariant ideal of an extended P-solvable loop can be computed.

\begin{proposition}[\cite{issac2017}]
  \label{prop:psolvableideal}
  For an extended P-solvable loop with program variables $v_1,\dots,v_m$ the
  invariant ideal is given by
  \[\Biggl(\Biggl(\sum_{j=1}^m
        \bigl\langle\smash{q_j(\fls{v}^{(0)}\kern-2pt,\fls{y})v^{(1)}_j}\kern-2pt
          -  \smash{p_j(\fls{v}^{(0)}}\kern-2pt,\fls{y}) \bigr\rangle\Biggr)\mathbin{:}\langle \operatorname{lcm}(q_1,\dots,q_m)\rangle^\infty + I(\fls{\theta},\fls{\zeta})\Biggr) \cap \K[\fls{v}^{(1)}\kern-2pt, \fls{v}^{(0)}].\]
\end{proposition}

\begin{example}
\newcommand{\lc}[1]{^{(#1)}\kern-2pt}
\newcommand{\lcc}[1]{^{(#1)}}

  Consider the following loop with relevant program variables $a,b$ and $c$.
  \begin{equation*}
    \begin{tabular}{ll}
      \WHILE\ true\ \DO\\
      \quad $a\ass  2\cdot (n+1)(n+\frac{3}{2})\cdot a$\\ 
      \quad $b\ass  4\cdot (n+1)\cdot b$\\ 
      \quad $c\ass  \frac{1}{2}\cdot (n+\frac{3}{2})\cdot c$\\ 
      \quad $n\ass n+1$\\
      \OD\\
    \end{tabular}
  \end{equation*}
  The extracted recurrence relations admit the following system of closed form
  solutions:
{ \setlength\abovedisplayskip{0pt plus 2pt minus 2pt}%
  \setlength\belowdisplayskip{8pt plus 2pt minus 2pt}%
  \begin{align*}
    a_n &= 2^n\cdot a_0\cdot\fn{n}\cdot\fn{n+\frac{1}{2}},\\
    b_n &= 4^n\cdot b_0\cdot\fn{n},\\
    c_n &= 2^{-n}\cdot c_0\cdot\fn{n+\frac{1}{2}}.
  \end{align*}
  }Since every closed form solution is of the form~(\ref{eq:psolvable}) we have
  an extended P-solvable loop, and we can apply
  Proposition~\ref{prop:psolvableideal} to compute the invariant ideal:
  %
  \begin{align*}
      (\mathcal{I} + I(\fls{\theta},\fls{\zeta})) \cap \K[a\lc{1},b\lc{1},c\lc{1},a\lc{0},b\lc{0},c\lcc{0}]
    = \langle b\lc{1} \cdot c\lc{1} \cdot a\lc{0} - a\lc{1} \cdot b\lc{0} \cdot c\lcc{0} \rangle,
  \end{align*}
  %
  where
  %
  \begin{alignat*}2
    &\mathcal{I} &\;=\;& \langle a\lc{1} - y_1\cdot a\lc{0} \cdot z_1 z_2, b\lc{1} - y_2\cdot b\lc{0} \cdot z_1, c\lc{1} - y_3\cdot c\lc{0} \cdot z_2 \rangle,\\
    &I(\fls{\theta},\fls{\zeta}) &\;=\;& \langle y_1^2 - y_2, y_1 y_3 - 1, y_2 y_3 - y_1 \rangle.
  \end{alignat*}
  %
  The ideal $I(\fls{\theta},\fls{\zeta})$ in variables $y_1,y_2,y_3$ is the
  set of all algebraic dependencies among $2^n,4^n$ and $2^{-n}$, and
  $\mathcal{I}$ is generated by the closed form solutions where exponential and
  factorial sequences are replaced by variables $y_1,y_2,y_3$ and $z_1,z_2$.
\end{example}


\subsection{Algebraic Dependencies of Composed Rational Functions with Side Conditions}
\label{sec:dependencies}

In this section we give the prerequisites for proving termination of the
invariant generation method for multi-path loops (Section~\ref{sec:loopcond}).
The results of this section will allow us to proof termination by applying
Theorem~\ref{thm:finitedimension}.

Let $\fls{v}^{(i)} = v_1^{(i)},\dots,v_m^{(i)}$ and
$\fls{y}^{(i)} = y_1^{(i)},\dots,y_\ell^{(i)}$ for $i\in\set N$. We model the
situation in which the value of the $j$th loop variable after the execution of
the $i$th block in~\eqref{eq:program} is given by a rational function in the
$\fls{y}^{(i)}$ (which, for us, will be the exponential and factorial sequences
as well as the loop counter) and the \lq old\rq\ variable values
$\fls{v}^{(i-1)}$ and is assigned to $v_j^{(i)}$. Set
${\I_0=\sum_{j=1}^m\langle v^{(1)}_j\kern-2pt-v^{(0)}_j\rangle}$ and let
${I_i\ideal\K[\fls{y}^{(i)}]}$ for $i\in\set N^*$. Furthermore, let 
$q^{(i)}_j,p^{(i)}_j\kern-2pt\in\K[\fls{v}^{(i)}\kern-2pt,\fls{y}^{(i)}]$ such
that for fixed~$i$ there exists a $\fls{y}\in V(I_i)$ with
${p_j^{(i)}(\fls{v}^{(i)}\kern-2pt,\fls{y})/q_j^{(i)}(\fls{v}^{(i)}\kern-2pt,\fls{y})=\fls{v}^{(i)}_j\kern-2pt}$
for all $j$ and with $d_i:=\operatorname{lcm}(q^{(i)}_1,\dots,q^{(i)}_m)$ we have $d_i\notin
I_i$ and $d_i(\fls{v_i},\fls{y})=1$. Set {
  \setlength\abovedisplayskip{5pt plus 2pt minus 2pt}%
  \setlength\belowdisplayskip{5pt plus 2pt minus 2pt}%
\[J_i=\sum_{j=1}^{m}\langle
q^{(i)}_j(\fls{v}^{(i)}\kern-2pt,\fls{y}^{(i)})v^{(i+1)}_j\kern-2pt-p^{(i)}_j(\fls{v}^{(i)}\kern-2pt,\fls{y}^{(i)})\rangle.\]\vspace{-10pt}
}\begin{remark}
  The requirement for the existence of a point $\fls{y}$ in $V(I_i)$ such that
${p_j^{(i)}(\fls{v}^{(i)}\kern-2pt,\fls{y})/q_j^{(i)}(\fls{v}^{(i)}\kern-2pt,\fls{y})=\fls{v}^{(i)}_j\kern-2pt}$ for all
$j$ and $d_i(\fls{v_i},\fls{y})=1$ is always fulfilled in our context, as it is a formalization of the fact
that the execution of a loop $L^*$ also allows that it is
executed zero times, meaning the values of the program variables do not
change. 
\end{remark}

In order to develop some intuition about the following, consider a list of
consecutive loops $L_1;L_2;L_3;\dots$ where each of them is extended P-solvable.
Intuitively, the ideals $I_i$ then correspond to the ideal of algebraic
dependencies among the exponential and factorial sequences occurring in $L_i$,
whereas $J_i$ stands for the ideal generated by the closed form solutions of
$L_i$. Moreover, the variables $v_j^{(i+1)}$ correspond to the values of the
loop variables after the execution of the loop $L_i$. The following iterative
computation then allows us to generate the invariant ideal for
$L_1;L_2;L_3;\dots$
\[\I_i:=((J_i+\I_{i-1}+I_i)\mathbin{:}\langle
d_i\rangle^\infty)\cap \K[\fls{v}^{(i+1)}\kern-2pt,\fls{v}^{(0)}] \]
Now the remaining part of this section is devoted to proving properties of the
ideals $\I_i$ which will help us to show that there exists an index $k$
such that $\I_k = \I_{k'}$ for all $k' > k$ for a list of consecutive loops
$L_1;\dots;L_r;L_1;\dots;L_r;\dots$ with $r\in\N$.

First note that the ideal $\I_i$ can be rewritten as
\begin{align}
\label{eq:ideal}
  \I_{i}=\{p&\in\smash{\K[\fls{v}^{(i+1)}\kern-2pt,\fls{v}^{(0)}]}\mid \exists
              q\in \I_{i-1},k\in\set N: \notag\\ 
   & q\equiv \smash{d_i^kp(r^{(i)}_1(\fls{v}^{(i)}\kern-2pt,\fls{y}^{(i)}),\dots,r^{(i)}_m(\fls{v}^{(i)}\kern-2pt,\fls{y}^{(i)}),\fls{v}^{(0)})}
  \;\;(\operatorname{\mathbf{mod}}\; I_i)\}.
\end{align}
If $I_i$ is radical, an equation $\mathbf{mod}\; I_i$ is, informally speaking, the
same as substituting~$\fls{y}$ with values from $V(I_i)$, so~\eqref{eq:ideal} translates to
\begin{align}
\label{eq:radical}
  \I_{i}=\{p&\in\smash{\K[\fls{v}^{(i+1)}\kern-2pt,\fls{v}^{(0)}]} \mid \exists q\in \I_{i-1},k\in\set N: \notag\\
  &\forall\fls{y}\in V(I_i) :
  q=\smash{d_i^kp(r^{(i)}_1(\fls{v}^{(i)}\kern-2pt,\fls{y}),\dots,r^{(i)}_m(\fls{v}^{(i)}\kern-2pt,\fls{y}),\fls{v}^{(0)})}\}.
\end{align}

We now get the following subset relation between two consecutively computed ideals $\I_i$.

\begin{lemma}
\label{lem:radical} If $I_i$ is radical, then $\I_i \subseteq \I_{i-1}|_{\fls{v}^{(i-1)}\leftarrow \fls{v}^{(i)}}$.
\end{lemma}

\begin{proof}
  Let $p\in\I_i$. We have to show that there is an $r\in\I_{i-2}$ and a
  $k\in\set N$ such that 
  \[r\equiv
    d_{i-1}^k\smash{p(r_1^{(i-1)}(\fls{v}^{(i-1)}\kern-2pt,\fls{y}^{(i-1)}),\dots,r^{(i-1)}_m(\fls{v}^{(i-1)}\kern-2pt,\fls{y}^{(i-1)}),\fls{v}^{(0)})}\;\;(\operatorname{\mathbf{mod}}\;
    I_{i-1}).\]
  Since $I_i$ is radical, there is a $q\in \I_{i-1}$, a $z\in\set N$, and a $\fls{y}\in V(I_i)$
  with
  \[q=d_i^z\smash{p(r^{(i)}_1(\fls{v}^{(i)}\kern-2pt,\fls{y}),\dots,r^{(i)}_m(\fls{v}^{(i)}\kern-2pt,\fls{y}),\fls{v}^{(0)})
    = p(\fls{v}^{(i)}\kern-2pt,\fls{v}^{(0)})}.\]
  Then, by Equation~\eqref{eq:ideal} for $\I_{i-1}$, there is an $r\in \I_{i-2}$ with the desired property.
\end{proof}

For prime ideals, we get an additional property:
\begin{lemma}
\label{lem:prime}
  If $\I_{i-1}$ and $I_i$ are prime, then so is $\I_i$.
\end{lemma}

\begin{proof}
  Let $a\cdot b\in \I_i$ and denote by $a|_r$ and $b|_r$ the rational functions where each $v_j^{(i+1)}$ is
  substituted by $r^{(i)}_j$  in $a,b$ respectively. Then there is a $q\in
  \I_{i-1}$ and a $k=k_1+k_2\in\set N$ with $d_i^{k_1}a|_r,d_i^{k_2}b|_r\in\K[\fls{v}^{(i+1)}\kern-2pt,\fls{v}^{(0)}]$
  \[q\equiv d_i^k(a\cdot b)|_{r} \equiv d_i^{k_1}a|_{r}\cdot d_i^{k_2}b|_{r}
    \;\;(\operatorname{\mathbf{mod}}\; I_i)\]
  If $d_i^ka|_r$ is zero modulo $I_i$, then $a$ is an element of $\I_i$, as $0\in\I_{i-1}$. The same
  argument holds for $b$. Suppose that
  $d_i^{k_1}a|_r,d_i^{k_2}b|_r\not\equiv 0\;(\operatorname{\mathbf{mod}}\; I_i)$.  Then, since
  $I_i$ is prime, $\K[\fls{y}^{(i)}]/I_i$ is an integral domain, and so it follows that
  $q\not\equiv 0\;(\operatorname{\mathbf{mod}}\; I_i)$. Now, because
  $\I_{i-1}$ is prime, it follows without loss of generality that
  $d_i^{k_1}a|_r\in\I_{i-1}$, from which we get $a\in\I_i$.
\end{proof}

We now use Lemmas~\ref{lem:radical} and~\ref{lem:prime} to give details about
the minimal decomposition of $\I_i$.

\begin{proposition}
\label{prop:mindec}
  For fixed $i_0\in\set N$, let all $I_i$, $0\leq i \leq i_0$ be radical and let
  $\I_{i_0}=\smash{\bigcap_{k=0}^{n}} P_k$ be the minimal decomposition of
  $\I_{i_0}$. Then
  \begin{enumerate}
  \item for each $k$ there exist prime ideals $I_{k,1},I_{k,2},\dots$ such that $P_k$
    is equal to a $\I_{k,{i_0}}$ constructed as above with $J_1,\dots,J_{i_0}$ and
    $I_{k,1},\dots,I_{k,{i_0}}$. 
  \item if $I_{i_0+1}$ is radical and
    $\I_{{i_0}+1}=\smash{\bigcap_{j=0}^{n'}} P'_j$ is the minimal decomposition
    of $\I_{i+1},$ then, for each $P'_j$ there exists a $P_k$ such that
    $P'_j\subseteq P_k|_{\fls{v}^{(i_0)}\leftarrow \fls{v}^{(i_0+1)}}$.
  \end{enumerate}
\end{proposition}

\begin{proof}
  We prove 1.\ by induction. For $i_0=0$, there is nothing to show. Now assume the
  claim holds for some $i_0\in\set N$ and let
  $I_{i_0+1}=\smash{\bigcap_{j=0}^w} Q_j$ be the minimal decomposition of
  $I_{i_0+1}$. With this we get
\begin{align*}
  \I_{i_0+1} & =  (J_{i_0+1}+\I_{i_0}+I_{i_0+1})\mathbin{:}\langle
d_{i_0+1}\rangle^\infty\cap \K[\fls{v}^{(i_0+1)}\kern-2pt,\fls{v}^{(0)}]\\[6pt]
  & =  \left(\bigcap_{k=0}^n J_{i_0+1} + P_k + \bigcap_{j=0}^w Q_j\right)\mathbin{:}\langle
d_{i_0+1}\rangle^\infty\cap \K[\fls{v}^{(i+1)}\kern-2pt,\fls{v}^{(0)}]\\
  & =  \raisebox{-1.5pt}{$\Biggl($}\bigcap_{k=0}^n\bigcap_{j=0}^w \underbrace{(J_{i_0+1}+P_k+Q_j)\mathbin{:}\langle
d_{i_0+1}\rangle^\infty\cap
    \K[\fls{v}^{(i_0+1)}\kern-2pt,\fls{v}^{(0)}]}_{\tilde{I}_{k,j}}\raisebox{-1.5pt}{$\Biggr)$}.
\end{align*}
By the induction hypothesis, each $P_k$ admits a construction as above, and thus
so does $\tilde{I}_{k,j}$. By Lemma~\ref{lem:prime}, $\tilde{I}_{k,j}$ is prime.
This shows 1. The second claim then follows from the fact that the prime ideals
in the minimal decomposition of $\I_{i_0+1}$ are obtained from the $P_k$ via
$J_{i_0+1}$ and $Q_j$. Since the $Q_j$ are prime, they are also radical, and the
claim follows from Lemma~\ref{lem:radical}.
\end{proof}


\subsection{Loops with conditional branches}\label{sec:loopcond}

In this section, we extend the results of Section~\ref{sec:loopass} to loops
with conditional branches. Without loss of generality, we define our algorithm
for a loop of the form
\[\WHILE\ \dots\ \DO\ L_1;L_2;\dots;L_r\ \OD\] where $L_i = B_i^*$ and $B_i$ is
a block containing variable assignments only.

Let $I(\fls{\theta}_i,\fls{\zeta}_i)$ denote the ideal of all algebraic
dependencies as described in Section~\ref{sec:loopass} for a inner loop $L_i$.
As every inner loop provides its own loop counter, we have that the exponential
and factorial sequences of distinct inner loops are algebraically independent.
Therefore $I(\fls{\theta},\fls{\zeta}) := \sum_{i=0}^r
I(\fls{\theta}_i,\fls{\zeta}_i)$ denotes the set of all algebraic dependencies
between exponential and factorial sequences among the inner loops
$L_1,\dots,L_r$. 


Consider loop bodies $B_1,\dots,B_r$ with common loop variables $v_1,\dots,v_m$.
Suppose the closed form of $v_j$ in the $i$th loop body is given by a rational function
in $m+k+\ell+1$ variables:
\[\smash{v^{(i+1)}_j}\kern-2pt =  \smash{r^{(i)}_j(\fls{v}^{(i)}}\kern-2pt,\fls{\theta}^n,(n+\fls{\zeta})^{\underline{n}},n),\]
where $\smash{v^{(i)}_j}\kern-2pt$ and $\smash{v^{(i+1)}_j}\kern-2pt$ are
variables for the value of $v_j$ before and after the execution of the loop
body. Then we can compute the ideal of all polynomial invariants of the 
non-deterministic program $(B_1^*;B_2^*;\dots;B_r^*)^*$ with 
Algorithm~\ref{alg:polyinv-mutlipath}.

\begin{algorithm}
  \caption{Invariant generation via fixed point computation}
  \label{alg:polyinv-mutlipath}
  \begin{algorithmic}[1]
    \vskip.5\baselineskip
    \Require{Loop bodies $B_1,\dots,B_r$ as described.}
    \Ensure{The ideal of all polynomial invariants of $(B_1^*;B_2^*;\dots;B_r^*)^*$.}
    \algrule
    \State Compute $I:=I(\fls{\theta},\fls{\zeta})$ as described above \label{line:algdep}
    \State $\I_{old}=\{0\}$, $\I_{new}=\smash{\sum_{j=1}^m\langle v^{(1)}_j-v^{(0)}_i\rangle}$, $j=0$
    \While{$\I_{old}|_{\fls{v}^{((j-1)\cdot r + 1)}\leftarrow \fls{v}^{(j\cdot r+1)}}\neq \I_{new}$~\algorithmicand~$\I_{new}\neq \{0\}$}\label{line:loop}
      \State $\I_{old} \gets \I_{new}$, $j \gets j+1$
      \For{$i=1,\dots,r$}
        \State $\I_{new} \gets (J_{i\cdot j}+\I_{old}+I)\cap \K[\fls{v}^{(i\cdot j+1)},\fls{v}^{(0)}]$
      \EndFor
    \EndWhile
    \State \Return $\I_{new}$
  \end{algorithmic}
\end{algorithm}
\goodbreak
\begin{lemma}
\label{lem:seq}
  $I(\fls{\theta},\fls{\zeta})$ is a radical ideal.
\end{lemma}

\begin{proof}
  The elements of $I(\fls{\theta})$ represent C-finite sequences,
  i.e.\ sequences of the form
  \[f_1(n)\theta_1^n+\dots+f_k^n\theta_k^n,\]
  for univariate polynomials $f_1,\dots,f_k\in\set K[y_0]$ and pairwise distinct
  $\theta_1,\dots,\theta_k\in\set K$.
  The claim is then proven by the fact that the Hadamard-product
  $a^2(n,a(0))$ of a C-finite sequence $a(n,a(0))$ with itself is zero if and
  only if $a(n,a(0))$ is zero, and $I(\fls{\theta},\fls{\zeta})$ is the
  extension of $I(\fls{\theta})$ to $\set K[y_0,\dots,y_{k+\ell}]$.
\end{proof}

\begin{theorem}
  \label{thm:algo}
  Algorithm~\ref{alg:polyinv-mutlipath} is correct and terminates.
\end{theorem}

\begin{proof} 
  The algorithm iteratively computes the ideals $\I_1,\I_2,\dots$ as in
  Section~\ref{sec:dependencies}, so we will refer to $I_{old}$ and $I_{new}$ as
  $\I_i$ and $\I_{i+1}$.

  \textit{Termination:} $\I_0$ is a prime ideal of height $m$. Suppose after an
  execution of the outer loop, the condition $\I_{i}|_{\fls{v}^{(i)}\leftarrow
  \fls{v}^{(i+1)}}\neq \I_{i+1}$ holds. As $I(\fls{\theta},\fls{\zeta})$ is
  radical by Lemma~\ref{lem:seq}, we then get
  $\I_{i+1}\subset\I_{i}|_{\fls{v}^{(i)}\leftarrow \fls{v}^{(i+1)}}$ by
  Lemma~\ref{lem:radical}. Thus there is a
  ${p\in\K[\fls{v}^{(i+1)}\kern-2pt,\fls{v}^{(0)}]}$ with
  $p\in\I_{i}|_{\fls{v}^{(i)}\leftarrow \fls{v}^{(i+1)}}$ and $p\notin\I_{i+1}$.
  Then, by Proposition~\ref{prop:mindec}, all prime ideals $P_k$ in the minimal
  decomposition of $\I_{i+1}$ are have to be subsets of the prime ideals in the
  minimal decomposition of $\I_{i}|_{\fls{v}^{(i)}\leftarrow \fls{v}^{(i+1)}}$,
  where at least one of the subset relations is proper. Since $p\notin
  \I_{i+1}$, the height of at least one $P_k$ has to be reduced. The height of
  each prime ideal is bounded by the height of $\I_0$.

  \textit{Correctness:} Let $i\in\N$ be fixed and denote by
  $I(B;i)\ideal\K[\fls{v}^{(i+1)},\fls{v}^{(0)}]$ the ideal of all polynomial
  invariants for the non-deterministic program
  \[(B_1^*;\dots;B_r^*)^{\sfrac{i}{r}};B_1^*;\dots;B_{i\;\mathrel{\operatorname{\mathbf{rem}}}\;r}^*.\] It suffices to
  show that $\I_i$ is equal to $I(B;i)$. In fact, after $i_0$ iterations with
  $\I_{i_0}=\I_{i_0+1}=\I_{i_0+2}=\dots$, it follows that $\I_{i_0}$ is the
  ideal of polynomial invariants for $(B_1^*;\dots;B_r^*)^*\kern-2pt.$ Let $p\in
  I(B;i)$. The value of the program variable $v_j$ in the program
  $B_1^*;\dots;B_{i\;\mathrel{\operatorname{\mathbf{rem}}}\;r}^*$ is given as the value
  of a composition of the closed forms of each $B_k$:
  %
  \[v_j=
  \smash{p^{(i)}_j\bigg(p^{(i-1)}\Big(\dots\big(p^{(1)}(\fls{v}^{(0)}\kern-2pt,\fls{s}_{n_1}),\dots\big),\fls{s}_{n_{i-1}}\Big),\fls{s}_{n_i}\bigg)},\]
  with $\fls{s}_n = n,\fls{\theta}^n,(n + \fls{\zeta})^{\underline{n}}$ and
  $n_1,\dots,n_i\in\set N$. The correctness then follows from the fact that that
  $\I_i$ is the ideal of all such compositions under the side condition that
  $(\fls{\theta}^n,(n+\fls{\zeta})^{\underline{n}},n)\in V(I(\fls{\theta},\fls{\zeta}))$ for
  any $n\in\set N$.
\end{proof}

Revisiting the subset relations of the prime ideals in the minimal decomposition
of $\I_0,\I_1,\dots$ gives an upper bound for the necessary number of iterations
in the algorithm.

\begin{corollary}\label{cor:bound}
  Algorithm~\ref{alg:polyinv-mutlipath} terminates after at most $m$ iterations
  of the {while-loop} at line~\ref{line:loop}.
\end{corollary}

\begin{proof}
  Suppose the algorithm terminates after $k_0$ iterations of the outer loop. We
  look at the ideals $\I_{r\cdot k}$, $k\in\{0,\dots,k_0\}$.  For a prime ideal
  $P$ in the minimal decomposition of any $\I_{r\cdot (k+1)}$, there is a prime
  ideal $Q$ in the minimal decomposition of $\I_{r\cdot k}$ such that $P\subseteq
  Q$. If $P=Q$, then $P$ is a prime ideal in the minimal decomposition of each
  $\I_{r\cdot (k')}$, $k'>k$. This holds because there are only $r$ many $J_i$.
  So if $Q$ does not get replaced by smaller prime ideals in $\I_{r\cdot
  k+1},\I_{r\cdot k+2}\dots,\I_{r\cdot (k+1)}$, it has to be part of the minimal
  decomposition for any subsequent $\I_i$. From this it follows that, for each
  $k$, there is a prime ideal $P_k$ in the minimal decomposition in $\I_{r\cdot
  k}$, such that $P_0\supset P_1\supset\dots\supset P_{k_0}$ is a chain of
  proper superset relations, which then proves the claim since the height of
  $P_0=\I_0$ is $m$.
\end{proof}

\begin{example}\label{ex:euclidex}
  Consider a multi-path loop $L$
  \[\WHILE\ \dots\ \DO~L_1;L_2~\OD\]
  containing the following nested loops $L_1$ and $L_2$ and the corresponding
  closed form solutions:
  \begin{center}
  \begin{minipage}{.22\columnwidth}
    \center
    \begin{tabular}{l}
      \WHILE\ \dots\ \DO\\
      \quad $a\ass  a - b$\\ 
      \quad $p\ass  p - q$\\ 
      \quad $r\ass  r - s$\\
      \OD\\
    \end{tabular}
  \end{minipage}
  \begin{minipage}{.22\columnwidth}
    \center
    \begin{tabular}{l}
      \\
      $a_n = a_0 - nb_0$\\ 
      $p_n = p_0 - nq_0$\\ 
      $r_n = r_0 - ns_0$\\ 
      \\
    \end{tabular}
  \end{minipage}
  \qquad
  \begin{minipage}{.22\columnwidth}
    \center
    \begin{tabular}{l}
      \WHILE\ \dots\ \DO\\
      \quad $b\ass  b - a$\\ 
      \quad $q\ass  q - p$\\ 
      \quad $s\ass  s - r$\\
      \OD\\
    \end{tabular}
  \end{minipage}
  \begin{minipage}{.22\columnwidth}
    \center
    \begin{tabular}{l}
      \\
      $b_m = b_0 - ma_0$\\ 
      $q_m = q_0 - mp_0$\\ 
      $s_m = s_0 - mr_0$\\ 
      \\
    \end{tabular}
  \end{minipage}
  \end{center}

For simplicity we chose inner loops without algebraic dependencies, i.e.~$I$ at
line~\ref{line:algdep} will be the zero ideal and we therefore neglect it in the
following computation. Moreover, we write $a_i$ instead of $a^{(i)}$. We start
with
{%
\setlength\abovedisplayskip{5pt}%
\setlength\belowdisplayskip{5pt}%
\[ \I_0 = \langle a_1-a_0, b_1-b_0, p_1-p_0, q_1-q_0, r_1-r_0, s_1-s_0  \rangle \] 
}%
followed by the first loop iteration: 
{%
\setlength\abovedisplayskip{5pt}%
\setlength\belowdisplayskip{5pt}%
\begin{align*}
  \I_1 &=(J_1 + \I_0) \cap \K[a_0,b_0,p_0,q_0,r_0,s_0,a_2,b_2,p_2,q_2,r_2,s_2]\\
    &= \langle b_0-b_2,q_0-q_2,s_0-s_2,-p_0 s_2+p_2 s_2+q_2 r_0-q_2 r_2,\\
    &\qquad a_0 s_2-a_2 s_2-b_2 r_0+b_2 r_2,a_0 q_2-a_2 q_2-b_2 p_0+b_2 p_2 \rangle 
\end{align*}
}%
where
{%
\setlength\abovedisplayskip{5pt}%
\setlength\belowdisplayskip{5pt}%
\[ J_1 = \langle a_2-a_1+b_1 n,p_2 - p_1 + q_1 n,r_2 - r_1 + s_1 n,b_2-b_1,q_2-q_1,s_2-s_1 \rangle \]
}%
The following ideal $\I_2$ is then the invariant ideal for the first iteration
of the outer loop $L$.
{%
\setlength\abovedisplayskip{5pt}%
\setlength\belowdisplayskip{5pt}%
  \begin{align*}
    \I_2 &=(J_2 + \I_1) \cap \K[a_0,b_0,p_0,q_0,r_0,s_0,a_3,b_3,p_3,q_3,r_3,s_3]\\
      &= \langle-p_0 r_3 s_0+p_3 r_3 s_3+p_3 r_0 s_0-p_3 r_0 s_3-q_3 r_3^2+q_3 r_0 r_3,\\
      &\qquad -p_3 s_0+p_3 s_3+q_0 r_3-q_3 r_3,  -p_0 s_0+p_3 s_3+q_0 r_0-q_3 r_3, \\
      &\qquad a_3 s_0-a_3 s_3-b_0 r_3+b_3 r_3, a_0 q_0-a_3 q_3-b_0 p_0+b_3 p_3,\\
      &\qquad a_3 p_0 s_3-a_3 p_3 s_3-a_3 q_3 r_0+a_3 q_3 r_3-b_0 p_3 r_0+b_3 p_3 r_0+b_0 p_0 r_3-b_3 p_0 r_3,\\
      &\qquad a_3 q_0-a_3 q_3-b_0 p_3+b_3 p_3, a_0 s_0-a_3 s_3-b_0 r_0+b_3 r_3,\\
      &\qquad -a_0 p_3 s_3+a_3 p_3 s_3+a_0 q_3 r_3-a_3 q_3 r_3+b_0 p_3 r_0-b_0 p_0 r_3,\\
      &\qquad -a_3 b_0 r_0+a_3 b_3 r_3+a_0 b_0 r_3-a_0 b_3 r_3-a_3^2 s_3+a_0 a_3 s_3,\\
      &\qquad -a_3 b_0 p_0+a_3 b_3 p_3+a_0 b_0 p_3-a_0 b_3 p_3-a_3^2 q_3+a_0 a_3 q_3\rangle
  \end{align*}
}%
where
{%
\setlength\abovedisplayskip{5pt}%
\setlength\belowdisplayskip{5pt}%
\[ J_2 = \langle b_3 -b_2+ a_2 m,q_3 -q_2 + p_2 m,s_3 -s_2 + r_2 m,a_3-a_2,p_3-p_2,r_3-r_2 \rangle \]
}%
By continuing this computation we get the following ideals $\I_4$ and $\I_6$
which are the invariant ideals after two and three iterations of the outer loop
$L$ respectively.
%
%
\begin{align*}
  \I_4 &= \langle p_0 s_0 - p_5 s_5 - r_0 q_0 + r_5 q_5, \\
    &\qquad b_5 p_5 - b_0 p_0 + a_0 q_0 - a_5 q_5, \\
    &\qquad b_5 r_5 - b_0 r_0 + a_0 s_0 - a_5 s_5, \\
    &\qquad b_5 (-p_5 s_0 + r_5 q_0) + b_0 (p_5 s_5 - r_5 q_5) + a_5 (-s_5 q_0 + s_0 q_5),  \\
    &\qquad b_5 (-p_5 r_0 + p_0 r_5) + a_5 (-p_0 s_5 + r_0 q_5) +  a_0 (p_5 s_5 - r_5 q_5), \\  
    &\qquad b_0 p_0 (-p_5 s_5 + r_5 q_5) + b_5 (p_5^2 s_5 - p_0 r_5 q_0 + p_5 (r_0 q_0 - r_5 q_5)) + {}\\
    &\qquad\qquad a_5 (p_0 s_5 q_0 + q_5 (-p_5 s_5 - r_0 q_0 + r_5 q_5))\rangle
\end{align*}
\begin{align*}
  \I_6 &= \langle p_0 s_0 - p_7 s_7 - r_0 q_0 + r_7 q_7, \\
    &\qquad b_7 p_7 - b_0 p_0 + a_0 q_0 - a_7 q_7, \\
    &\qquad b_7 r_7 - b_0 r_0 + a_0 s_0 - a_7 s_7, \\
    &\qquad b_7 (-p_7 s_0 + r_7 q_0) + b_0 (p_7 s_7 - r_7 q_7) + a_7 (-s_7 q_0 + s_0 q_7),  \\
    &\qquad b_7 (-p_7 r_0 + p_0 r_7) + a_7 (-p_0 s_7 + r_0 q_7) +  a_0 (p_7 s_7 - r_7 q_7), \\  
    &\qquad b_0 p_0 (-p_7 s_7 + r_7 q_7) + b_7 (p_7^2 s_7 - p_0 r_7 q_0 + p_7 (r_0 q_0 - r_7 q_7)) + {}\\
    &\qquad\qquad a_7 (p_0 s_7 q_0 + q_7 (-p_7 s_7 - r_0 q_0 + r_7 q_7))\rangle
\end{align*}
Note that we now reached the fixed point as $\I_6 =
\I_4|_{\fls{v}^{(5)}\leftarrow \fls{v}^{(7)}}$.
\end{example}

Corollary~\ref{cor:bound} provides a bound on the number of iterations in
Algorithm~\ref{alg:polyinv-mutlipath}. Therefore, we know at which stage we have
to reach the fixed point of the computation at the latest, viz.~after computing
$\I_{r\cdot m}$. This fact allows us to construct a new algorithm which computes
the ideal $\I_{r\cdot m}$ directly instead of doing a fixed point computation.
The benefit of Algorithm~\ref{alg:nofixedpoint} is that we have to perform only
one Gr\"obner basis computation in the end, although the new algorithm might
performs more iterations than Algorithm~\ref{alg:polyinv-mutlipath}.

\begin{algorithm}
  \caption{Invariant generation without fixed point computation}
  \label{alg:nofixedpoint}
  \begin{algorithmic}[1]
    \vskip.5\baselineskip
    \Require{Loop bodies $B_1,\dots,B_r$ as described.}
    \Ensure{The ideal of all polynomial invariants of $(B_1^*;B_2^*;\dots;B_r^*)^*$.}
    \algrule
    \State Compute $I:=I(\fls{\theta},\fls{\zeta})$ as described above
    \State $\I_{new}=\smash{\sum_{j=1}^m\langle v^{(1)}_j-v^{(0)}_i\rangle} + I$
    \For{$j = 1,\dots,m$}
      \For{$i=1,\dots,r$}
        \State $\I_{new} \gets (J_{i\cdot j}+\I_{new})$
      \EndFor
    \EndFor
    \State \Return $\I_{new} \cap \K[\fls{v}^{(m\cdot r+1)},\fls{v}^{(0)}]$
  \end{algorithmic}
\end{algorithm}

The proof of termination of the invariant generation method
of~\cite{completeinvariant} assumes that the ideal of algebraic dependencies is
prime. In general, this does not hold. Consider the following loop and its
closed forms with exponential sequences $2^n$ and $(-2)^n$:
\begin{center}
  \begin{minipage}{.35\columnwidth}
    \center
    \begin{tabular}{l}
      \WHILE\ \dots\ \DO\\
      \quad $x\ass  2 x$\\ 
      \quad $y\ass  -2 y$\\ 
      \OD\\
    \end{tabular}
  \end{minipage}
  \begin{minipage}{.35\columnwidth}
    \center
    \begin{tabular}{l}
      \\
      $x(n) = 2^n \cdot x(0)$\\ 
      $y(n) = (-2)^n \cdot y(0)$\\ 
      \\
    \end{tabular}
  \end{minipage}
\end{center}
The ideal of algebraic dependencies among the before-mentioned exponential
sequences is given by $\langle a^2 - b^2 \rangle$ which is obviously not prime.
As a consequence, the termination proof of~\cite{completeinvariant} is
incorrect. This paper closes this gap by providing a new algorithm and a
corresponding termination proof.


\section{Implementation and Experiments}\label{sec:implementation}

We implemented our method in the Mathematica package \Aligator\footnote{\Aligator requires the Mathematica
packages Hyper~\cite{hyper}, Dependencies~\cite{dependencies} and
FastZeil~\cite{fastzeil}, where the latter two are part of the compilation
package ErgoSum~\cite{ergosum}.}. \Aligator is open source and available at:
\begin{quote}\footnotesize\url{https://ahumenberger.github.io/aligator/}\end{quote}
%




\newcommand{\Fastind}{\textsc{Fastind}}
\newcommand{\Duet}{\textsc{Duet}}

\paragraph{\bf Comparison of generated invariants.} Based on the examples in
Figure~\ref{fig:loops} we show that our technique can infer invariants which
cannot be found by other state-of-the-art approaches. Our observations indicate
that our method is superior to existing approaches if the loop under
consideration has some \emph{mathematical meaning} like division or
factorization algorithms as depicted in Figure~\ref{fig:loops}, whereas the
approach of \cite{kincaidPOPL18} has advantages when it comes to programs with
complex flow.

The techniques of \cite{farewellgroebner} and \cite{kincaidPOPL18} were
implemented in tools called \Fastind{}\footnote{Available at
\url{http://www.irisa.fr/celtique/ext/polyinv/}} and \Duet{}\footnote{Available
at \url{https://github.com/zkincaid/duet}} respectively. 
Unlike \Aligator{} and \Fastind{}, \Duet{} is not a pure inference engine for
polynomial invariants, instead it tries to prove user-specified safety
assertions. In order to check which invariants can be generated by \Duet{}, we
therefore asserted the invariants computed by \Aligator{} and checked if \Duet{}
can prove them.

\begin{figure}
  \begin{subfigure}[b]{0.30\textwidth}
  \centering  
  \begin{tabular}{l}
    \WHILE\ $a \neq b$\ \DO\\
    \quad\IF\ $a > b$\ \THEN \\
      \qquad $a\ass  a - b$\\ 
      \qquad $p\ass  p - q$\\ 
      \qquad $r\ass  r - s$\\
    \quad\ELSE\\
      \qquad $b\ass  b - a$\\ 
      \qquad $q\ass  q - p$\\ 
      \qquad $s\ass  s - r$\\
    \quad\FI\\
    \OD\\
  \end{tabular}
  \caption{}
  \label{fig:euclidex}
  \end{subfigure}
  \quad
  %
  \begin{subfigure}[b]{0.30\textwidth}
  \centering  
  \begin{tabular}{l}
    \WHILE\ $r \neq 0$\ \DO\\
    \quad\IF\ $r > 0$\ \THEN \\
      \qquad $r\ass  r - v$\\ 
      \qquad $v\ass  v + 2$\\ 
    \quad\ELSE\\
      \qquad $r\ass  r + u$\\ 
      \qquad $u\ass  u + 2$\\ 
    \quad\FI\\
    \OD\\\\\\
  \end{tabular}
  \caption{}
  \label{fig:fermat}
  \end{subfigure}
  \quad
  %
  \begin{subfigure}[b]{0.30\textwidth}
  \centering
  \begin{tabular}{l}
    \WHILE\ $d \geq E$\ \DO\\
    \quad\IF\ $P < a + b$\ \THEN \\
      \qquad $b\ass  b / 2$\\ 
      \qquad $d\ass  d / 2$\\ 
    \quad\ELSE\\
      \qquad $a\ass  a + b$\\ 
      \qquad $y\ass  y + d / 2$\\ 
      \qquad $b\ass  b / 2$\\ 
      \qquad $d\ass  d / 2$\\ 
    \quad\FI\\
    \OD\\
  \end{tabular}
  \caption{}
  \label{fig:wensley}
  \end{subfigure}

  \caption{Three examples: (a) Extended Euclidean algorithm, (b) a variant of
  Fermat's factorization algorithm and (c) Wensley's algorithm for real
  division.}
  \label{fig:loops}
\end{figure}

Let us consider the loop depicted in Figure~\ref{fig:euclidex}. Since we treat
conditional branches as inner loops, we have that the invariants for this loop
are the same as for the loop in Example~\ref{ex:euclidex}. By instantiating the
generated invariants with the following initial values on the left we get the
following polynomial invariants on the right: 
%
\begin{center}
\vspace{-\baselineskip}
\begin{minipage}{0.3\textwidth}
\begin{align*}
    \hspace{5em} a_0 &\mapsto x \qquad \\
    b_0 &\mapsto y \qquad \\
    p_0 &\mapsto 1 \qquad \\
    q_0 &\mapsto 0 \qquad \\
    r_0 &\mapsto 0 \qquad \\
    s_0 &\mapsto 1
\end{align*}
\end{minipage}
\begin{minipage}{0.69\textwidth}
\begin{align}
  \tag{$I_1$} &1 + qr - ps               \\ 
  \tag{$I_2$} &bp - aq - y               \\
  \tag{$I_3$} &br - as + x               \\ 
  \tag{$I_4$} &-bp + aq - qry + psy      \\ 
  \tag{$I_5$} &br - as - qrx + psx       \\
  \tag{$I_6$} &(qr - ps)(-bp + aq + y)
\end{align}
\end{minipage}
\end{center}
Note that ($I_4$)-($I_6$) are just linear combinations of $(I_1)$-$(I_3)$.
However, \Fastind{} was able to infer $(I_1)$-$(I_3)$, whereas \Duet{} was only
able to prove $(I_2)$, $(I_5)$ and~$(I_6)$.

Other examples where \Aligator{} is superior in terms of the number of inferred
invariants are given by the loops in Figures~\ref{fig:fermat}
and~\ref{fig:wensley}. For Fermat's algorithm (Figure~\ref{fig:fermat}) and the
following initial values, \Aligator{} found one invariant, which was also found
by \Fastind{}. However, \Duet{} was not able to prove it.

\begin{center}
\vspace{-\baselineskip}
\begin{minipage}{.3\textwidth}
  \begin{align*}
    \hspace{4em}u_0 &\mapsto 2R + 1\\
    v_0 &\mapsto 1\\
    r_0 &\mapsto RR - N
  \end{align*}
\end{minipage}
\begin{minipage}{.69\textwidth}
  \begin{align*}
    \tag{$I_7$} &-4N - 4r - 2u + u^2 + 2v - v^2
  \end{align*}
\end{minipage}
\end{center}

In case of Wensley's algorithm (Figure~\ref{fig:wensley}) \Aligator{} was able
to identify the following three invariants. \Fastind{} inferred the first two
invariants, whereas \Duet{} could not prove any of them.
%
\begin{center}
\vspace{-\baselineskip}
\begin{minipage}{.3\textwidth}
  \begin{align*}
    \hspace{5em} a_0 &\mapsto 0 \\
    b_0 &\mapsto Q/2 \\
    d_0 &\mapsto 1 \\
    y_0 &\mapsto 0
  \end{align*}
\end{minipage}
\begin{minipage}{.69\textwidth}
  \begin{align*}
    \tag{$I_{8}$}&2b - dQ \\
    \tag{$I_{9}$}&ad - 2by \\
    \tag{$I_{10}$}&a - Qy
  \end{align*}
\end{minipage}
\vspace{-\baselineskip}
\end{center}



\noindent
\paragraph{\bf Benchmarks and Evaluation.}
For the experimental evaluation of our approach, we used the
following set of examples: 
(i)  18 programs taken from~\cite{farewellgroebner}; 
(ii) 4 new programs of extended
P-solvable loops that were created by us. All examples are available
at  the repository of \Aligator. 

Our experiments were performed on a machine with a 2.9 GHz Intel Core i5 and 16
GB LPDDR3 RAM; for each example, a timeout of $300$ seconds was set. When using
\Aligator{}, the Gr\"obner basis of the invariant ideal computed by {\Aligator}
was non-empty for each example; that is, for each example we were able to find
non-trivial invariants.

We evaluated \Aligator{} against \Fastind{}. As \Duet{} is not a pure inference
engine for polynomial invariants, we did not include it in the following
evaluation. When compared to~\cite{farewellgroebner}, we note that we do not fix
the degree of the polynomial invariants to be generated. Moreover, our method is
complete. That is, whenever \Aligator{} terminates, the basis of the polynomial
invariant ideal is inferred; any other polynomial invariant is a linear
combination of the basis polynomials.

\newcolumntype{R}{>{$}r<{$}}

\begin{table}
  \centering
  \caption{Experimental evaluation of \Aligator.}
  \label{tab:benchmarks}
  \def\arraystretch{1.08}
  \begin{subtable}{.39\textwidth}
  \caption{}
  \label{tab:singlepath}
  \begin{tabular}{|l|R|R|}
    \hline
    \textit{Single-path} & \Aligator & \Fastind{} \\
    \hline
    \texttt{cohencu}  & 0.072   & 0.043 \\\hline
    \texttt{freire1}  & 0.016   & 0.041 \\\hline
    \texttt{freire2}  & 0.062   & 0.048 \\\hline
    \texttt{petter1}  & 0.015   & 0.040 \\\hline
    \texttt{petter2}  & 0.026   & 0.042 \\\hline
    \texttt{petter3}  & 0.035   & 0.051 \\\hline
    \texttt{petter4}  & 0.042   & 0.104 \\\hline
    \texttt{petter5}  & 0.053   & 0.261 \\\hline
    \texttt{petter20} & 48.290  & 9.816 \\\hline
    \texttt{petter22} & 247.820 & 9.882 \\\hline
    \texttt{petter23} & TO      & 9.853 \\\hline
  \end{tabular}
  \end{subtable}
  \quad
  \begin{subtable}{.575\textwidth}    
  \caption{}
  \label{tab:multipath}  
  \begin{tabular}{|l|R|R|R|R|R|R|}
    \hline
    \textit{Multi-path} &  \#b & \#v & \#i &  \textsc{Al1} & \textsc{Al2} & \Fastind{}  \\\hline
    \texttt{divbin}     &   2   &  3  &  2  &   0.134    & 45.948  & 0.045  \\\hline
    \texttt{euclidex}   &   2   &  6  &  3  &   0.433    & TO      & 0.049  \\\hline
    \texttt{fermat}     &   2   &  3  &  2  &   0.045    & 0.060   & 0.043  \\\hline
    \texttt{knuth}      &   4   &  5  &  2  &   ~55.791   & TO      & 1.025  \\\hline
    \texttt{lcm}        &   2   &  4  &  3  &   0.051    & ~87.752  & 0.043  \\\hline
    \texttt{mannadiv}   &   2   &  3  &  2  &   0.022    & 0.025   & 0.048  \\\hline
    \texttt{wensley}    &   2   &  4  &  2  &   0.124    & 41.851  & err    \\\hline
    \texttt{extpsolv2}  &   2   &  3  &  2  &   0.192    & TO      & err    \\\hline
    \texttt{extpsolv3}  &   3   &  3  &  2  &   0.295    & TO      & err    \\\hline
    \texttt{extpsolv4}  &   4   &  3  &  2  &   0.365    & TO      & err    \\\hline
    \texttt{extpsolv10}~ &   10  &  3  &  2  &   0.951    & TO      & err    \\\hline
  \end{tabular}
  \end{subtable}
  \begin{tabular}{rcl}
    \\
    $\#b,\#v$ & $\dots$ & number of branches, variables \\
    $\#i$ & $\dots$ & number of iterations until fixed point reached \\
    $\textsc{Al1}$ & $\dots$ & \Aligator with Algorithm~\ref{alg:polyinv-mutlipath} (timeout $300s$) \\
    $\textsc{Al2}$ & $\dots$ & \Aligator with Algorithm~\ref{alg:nofixedpoint} (timeout $100s$) \\
    \Fastind{} & $\dots$ & OCaml version of the tool in~\cite{farewellgroebner}\footnotemark \\
    $TO,err$ & $\dots$ & timeout, error
  \end{tabular}
  \vspace{-\baselineskip}
\end{table}
\footnotetext{Testing the Maple implementation was not possible due to constraints regarding the Maple version.}

Table~\ref{tab:singlepath} summarizes our experimental results on single-path
loops, whereas Table~\ref{tab:multipath} reports on the results from multi-path
programs. The first column of each table lists the name of the benchmark. The
second and third columns of Table~\ref{tab:singlepath} report, on the timing
results of \Aligator{} and {\sc Fastind}, respectively. In
Table~\ref{tab:multipath}, the second column lists the number of branches
(paths) of the multi-path loop, whereas the third column gives the number of
variables used in the program. The fourth column reports on the number of
iterations until the fixed point is reached by
\Aligator{}, and hence terminates. The fifth and sixth columns,
labeled {\sc Al1} and {\sc Al2}, show the performance
of \Aligator{} when using Algorithm~\ref{alg:polyinv-mutlipath} or 
Algorithm~\ref{alg:nofixedpoint}, respectively. The last column of
Table~\ref{tab:multipath} lists the results obtained by {\sc
  Fastind}. In both tables, timeouts are denoted by $TO$, whereas
errors, due to the fact that the tool cannot be evaluated on the
respective example, are given as $err$.

The results reported in Tables~\ref{tab:singlepath} and~\ref{tab:multipath} show
the efficiency of \Aligator{}: in 14 out of 18 examples, \Aligator{} performed
significantly better than {\sc FastInd}. For the examples
\texttt{petter20}, \texttt{petter22} and \texttt{petter23}, the time-consuming
part in \Aligator{} comes from recurrence solving (computing the closed form of
the recurrence), and not from the Gr\"obner basis computation. We intend to
improve this part of \Aligator{} in the future. The examples 
\texttt{extpsolv2}, \texttt{extpsolv3}, \texttt{extpsolv4} and
\texttt{extpsolv10} are extended P-solvable loops with respectively 2,
3, 4, and 10 
nested conditional branches. The polynomial arithmetic of these
examples is not supported by \Fastind{}. The results of
\Aligator{} on  these examples indicate
that extended P-solvable loops do not increase the complexity of computing
the invariant ideal. 

We also compared the performance of \Aligator{} with
Algorithm~\ref{alg:polyinv-mutlipath} against Algorithm~\ref{alg:nofixedpoint}.
As shown in columns 5 and 6 of Table~\ref{tab:multipath},
Algorithm~\ref{alg:nofixedpoint} is not as efficient as
Algorithm~\ref{alg:polyinv-mutlipath}, even though
Algorithm~\ref{alg:nofixedpoint} uses only a single Gr\"obner basis computation.
We conjecture that this is due to the increased number of variables in the
polynomial system which influences the Gr\"obner basis computation. We therefore
conclude that several small Gr\"obner basis computations (with fewer variables)
perform better than a single large one.


\section{Conclusions}

We proposed a new  algorithm for computing the ideal of all polynomial
invariants for the class of extended P-solvable multi-path loops. The new
approach computes the invariant ideal for a non-deterministic program
$(L_1;\dots;L_r)^*$ where the $L_i$ are single-path loops. As a consequence, the
proposed method can handle loops containing (i) an arbitrary nesting of
conditionals, as these conditional branches can be transformed into a sequence
of single-path loops by introducing flags, and (ii) one level of nested
single-path loops.

Our method computes the ideals $\I_1,\I_2,\dots$ until a fixed point is reached
where $\I_i$ denotes the invariant ideal of $(L_1;\dots;L_r)^i$. This fixed
point is then a basis for the ideal containing all polynomial invariants for the
extended P-solvable loop. We showed that this fixed point computation is
guaranteed to terminate which implies the completeness of our method.
Furthermore, we gave a bound on the number of iterations we have to perform to
reach the fixed point. The proven bound is given by $m$ iterations where $m$ is
the number of loop variables.

We showed that our method can generate invariants which cannot be inferred by
other state-of-the-art techniques. In addition, we showcased the efficiency of
our approach by comparing our Mathematica package \Aligator{} with
state-of-the-art tools in invariant generation. 

Future research directions include the incorporation of the loop condition into
our method. So far we operate on an abstraction of the loop where we ignore the
loop condition and treat the loop as a non-deterministic program. By doing so we
might loose valuable information about the control flow of the program. By
employing $\mathrm{\Pi\Sigma^*}$-theory~\cite{carsten2} it might be possible to extend
our work also to loops containing arbitrary nesting of inner loops, which
reflects another focus for further research.


\vspace{1em}

\par\noindent
{\bf Acknowledgments.} We want to thank the anonymous reviewers for their
helpful comments and remarks.

\balance
\bibliographystyle{splncs03}
\bibliography{references}

\end{document}